\documentclass[aps,amsmath,onecolumn,amssymb]{revtex4}
\usepackage{amssymb}
\usepackage{amsthm}
\usepackage{graphicx}
\usepackage{dcolumn}
\usepackage{bm}
\usepackage{hyperref}
\usepackage{enumerate}
\usepackage{algorithm}
\usepackage{algpseudocode}

\newtheorem{theorem}{Theorem}
\newtheorem{lemma}{Lemma}
\newtheorem{definition}{Definition}
\newtheorem{corollary}{Corollary}
\input{Qcircuit.tex}

\usepackage[usenames,dvipsnames]{xcolor}
\hypersetup{
    colorlinks=true,       
    linkcolor=Maroon,          
    citecolor=OliveGreen,        
    filecolor=magenta,      
    urlcolor=Blue           
}

\begin{document}

\def\ket#1{|#1\rangle}
\def\bra#1{\langle#1|}
\newcommand{\ketbra}[2]{|#1\rangle\!\langle#2|}
\newcommand{\braket}[2]{\langle#1|#2\rangle}
\newcommand{\prob}[1]{{\rm Pr}\left(#1 \right)}
\newcommand{\expect}[2]{{\mathbb{E}_{#2}}\!\left\{#1 \right\}}
\newcommand{\var}[2]{{\mathbb{V}_{#2}}\!\left\{#1 \right\}}


\newcommand{\sde}{\mathrm{sde}}
\newcommand{\Z}{\mathbb{Z}}
\newcommand{\w}{\omega}
\newcommand{\K}{\kappa}

\newcommand{\Tchar}{$T$}
\newcommand{\T}{\Tchar~}
\newcommand{\ClT}{\{{\rm Clifford}, \Tchar\}~}
\newcommand{\Tcount}{\Tchar--count~}
\newcommand{\Tcountper}{\Tchar--count}
\newcommand{\Tcounts}{\Tchar--counts~}
\newcommand{\Tdepth}{\Tchar--depth~}
\newcommand{\Zr}{\Z[i,1/\sqrt{2}]}
\newcommand{\ve}{\varepsilon}

\newcommand{\eq}[1]{\hyperref[eq:#1]{(\ref*{eq:#1})}}
\renewcommand{\sec}[1]{\hyperref[sec:#1]{Section~\ref*{sec:#1}}}
\newcommand{\app}[1]{\hyperref[app:#1]{\ref*{app:#1}}}
\newcommand{\fig}[1]{\hyperref[fig:#1]{Figure~\ref*{fig:#1}}}
\newcommand{\thm}[1]{\hyperref[thm:#1]{Theorem~\ref*{thm:#1}}}
\newcommand{\lem}[1]{\hyperref[lem:#1]{Lemma~\ref*{lem:#1}}}
\newcommand{\tab}[1]{\hyperref[tab:#1]{Table~\ref*{tab:#1}}}
\newcommand{\cor}[1]{\hyperref[cor:#1]{Corollary~\ref*{cor:#1}}}
\newcommand{\defn}[1]{\hyperref[def:#1]{Definition~\ref*{def:#1}}}

\newenvironment{proofof}[1]{\begin{trivlist}\item[]{\flushleft\it
Proof of~#1.}}
{\qed\end{trivlist}}

\newcommand{\cu}[1]{{\textcolor{red}{#1}}}
\newcommand{\tout}[1]{{}}
\newcommand{\beq}{\begin{equation}}
\newcommand{\eeq}{\end{equation}}
\newcommand{\beqa}{\begin{eqnarray}}
\newcommand{\eeqa}{\end{eqnarray}}

\newcommand{\id}{\openone}
\title{Floating Point Representations in Quantum Circuit Synthesis}
\author{Nathan Wiebe$^{1,2}$, Vadym Kliuchnikov$^{1,3}$}
\address{$^1$ Institute for Quantum Computing, 200 University Ave West, Waterloo, ON, Canada}
\address{$^2$ Department of Combinatics \& Opt., University of Waterloo, Waterloo, ON, Canada}
\address{$^3$ Department of Computer Science, University of Waterloo, Waterloo, ON, Canada}

\begin{abstract}
We provide a non--deterministic quantum protocol that approximates the single qubit rotations $R_x(2\phi_1^2 \phi_2^2)$ using $R_x(2\phi_1)$ and $R_x(2\phi_2)$ and a constant number of Clifford and \T operations.  We then use this method to construct a ``floating point'' implementation of a small rotation wherein we use the aforementioned method to construct the exponent part of the rotation and also to combine it with a mantissa.  This causes the cost of the synthesis to depend more strongly on the relative (rather than absolute) precision required.  We analyze the mean and variance of the \Tcount required to use our techniques and provide new lower bounds for the $T$--count for ancilla free synthesis of small single--qubit axial rotations.  We further show that our techniques can use ancillas to beat these lower bounds with high probability.  We also discuss the \Tdepth of our method and see that the vast majority of the cost of the resultant circuits can be shifted to parallel computation paths.
\end{abstract}
\maketitle

\section{Introduction}

The ability to implement very small rotations is vitally important to quantum computation.  The ability to economically implement small rotations is essential for the quantum Fourier transform, which is an essential part of
Shor's factoring algorithm.  In quantum
computer simulations of local Hamiltonians (which encompasses simulations of quantum chemistry in second quantized form), the time evolution operator is simulated by breaking up the evolution time into a sequence of short time evolutions using Trotter--Suzuki formulas~\cite{lloyd,LA97, WML+10,KW+11,RWS12}.  The implementation of each timestep requires performing a single qubit $Z$--rotation through a very small angle.  In practice, rotations of $10^{-3}$ radians or smaller may be needed in order to ensure that upper bounds on the simulation error are appropriately small~\cite{RWS12}.  As the error tolerance shrinks for the simulation, these rotation angles must shrink as well.  This issue is problematic because existing algorithms for designing fault tolerant circuits to implement these small angle rotations can be very costly, both in the number of gates required and the classical computational time
required to find the appropriate gate sequences~\cite{CMT+09, UGS13}.

The Solovay--Kitaev theorem~\cite{DN06} is often used to estimate the cost of synthesizing the rotation gates using a finite gate library at cost polylogarithmic in the error tolerance.  Although polylogarithmic, the cost of performing gate synthesis using the Solovay--Kitaev theorem is polynomially greater than the lower bound of logarithmic scaling~\cite{HRC02}.
In recent months, great progress has been made to reduce the cost of synthesizing single qubit unitaries, and now methods for synthesizing these rotations have been proposed that are
polynomially more efficient than the Solovay--Kitaev theorem~\cite{KMM12,Sel12,vbasis}.  Another novel approach that has recently been proposed uses non--deterministic algorithms that consume pre--programmed ancilla states to perform these rotations~\cite{CJ12,DS12,CL13,fstdist}, rather than utilizing a complicated circuit synthesis method.  A major advantage of these ancilla assisted synthesis methods is that the resource states can be prepared before the algorithm is executed, substantially reducing the depth of the circuit and making the result more resilient to circuit failure; furthermore, any leftover states can also be used as resources in subsequent runs or even other quantum algorithms.  Such methods may be preferable to using traditional circuit synthesis methods in parallel quantum computation where
fast classical feed forward is available~\cite{CJ12}.

Our key innovation  is a quantum protocol that refines large $X$--rotations into smaller rotations.  In particular, given the ability to enact the rotations $R_x(2\phi_1)$ and $R_x(2\phi_2)$, our method provides a way to implement a rotation that is approximately $R_x(2\phi_1^2 \phi_2^2)$ if $\phi_1\phi_2\ll 1$.
We further show that, with high probability, this approach generates small single qubit rotations more efficiently than the best possible ancilla--free circuit synthesis method (using the \ClT gate library).  This is significant because it
not only shows that ancillas are a powerful resource for single qubit circuit synthesis, but also because it allows much more sophisticated computations to be performed on a rudimentary quantum computer.

This ability to generate small rotations and multiply the rotation angles of two
operations naturally opens the possibility of employing a ``floating point'' implementation of the rotation.
A floating point number is broken up into two parts: the mantissa and the exponent.  Both the mantissa and exponents are encoded as integers and they represent a number $\phi$ as $\phi=m \times 10^{e}$ where $m$ is the mantissa and $e$ is the exponent.  A major advantage of this representation is that extraneous digits of precision are not used to represent very small, or very large, numbers.  Our non--deterministic circuit can then be used to construct $e^{-i\phi X}$ by combining a mantissa unitary $U_m$ and an exponent unitary $U_e$.  For example, if $\phi\ll1$ then we could combine $U_m=e^{-i\sqrt{m} X}$ and $U_e=e^{-i (10^{e/2}) X}$ to approximate $e^{-i\phi X}$. We make this intuition precise in \sec{floatingpoint}.  We must emphasize that our approach does not conform to standard implementations of floating point arithmetic (such as IEEE 754) nor is either base $2$ or $10$  the natural base for the exponent in our synthesis technique; nonetheless, the approach is strongly analogous to floating point arithmetic.

Similar to floating point arithmetic, a major motivation for the use of floating point synthesis is that its cost depends more strongly on the relative precision needed for the rotation rather than the absolute precision, unlike traditional circuit synthesis techniques.  This is especially significant for quantum simulation because it is common for small rotations to appear that do not need to be implemented with high relative precision in such applications.  A further benefit of our approach is that the vast majority of the cost involves preparing resource states that are then consumed to perform the desired rotation.  These preparations can be performed offline and in parallel, which allows much of the cost to be shifted to parallel computational paths.  Finally, the approximant yielded by our method is \emph{precisely} an axial rotation meaning that the rotation yielded is of exactly the same form as the desired rotation.



Our paper is laid out as follows.  We introduce our non--deterministic circuit in \sec{gearbox} and show how to use it  recursively to generate small rotations in \sec{compgearbox} and compute the mean and the variance of the number of \T gates required to execute our circuits.  We then combine these ideas in \sec{floatingpoint} to produce the floating point representation of the desired rotation.  \sec{example} gives an example of floating point synthesis that shows that  it substantially reduces the number of \T gates needed to approximate the rotation $\exp(-i\pi Z/2^{16})\approx\exp(-i4.7937\times 10^{-5} Z)$ relative to \emph{optimal} ancilla--free synthesis.
  We finally show in~\sec{optimal} that our method for generating small single qubit rotations is more efficient than optimal circuit synthesis methods that are constrained to only use single qubit Clifford and \T gates and provide an explicit construction for this optimal synthesis method.

{
\begin{figure}[t]

\hspace{-1.5cm}
\begin{minipage}[t]{0.45\linewidth}
\centering
\[
\newcommand{\up}[1]{\push{\raisebox{6pt}{$#1$}}}
 \Qcircuit @C=0.7em @R=0.7em {
\lstick{\ket{0}}&\qw&\gate{U^1}&\ctrl{1}&\gate{U^{1\dagger}}&\qw&\meter\\
\lstick{\ket{0}}&\qw&\gate{U^2}&\ctrl{1}&\gate{U^{2\dagger}}&\qw&\meter\\
&&&\up{\vdots} \qwx[2] &&&\\
\lstick{\ket{0}}&\qw&\gate{U^d}&\ctrl{1}&\gate{U^{d\dagger}}&\qw&\meter\\
\lstick{\ket{\psi}}&\qw&\qw&\gate{-iX}&\qw&\qw&\qw\\
}
\]
\caption{Gearbox circuit $C^{(d)}(U^1,\ldots,U^d)$, which implements a small rotation on the input state $\ket{\psi}$ given that each measurement outcome is $0$.\label{fig:gearbox}}
\end{minipage}
\hspace{.35cm}
\begin{minipage}[t]{0.45\linewidth}
\centering
\[
\newcommand{\up}[1]{\push{\raisebox{6pt}{$#1$}}}
 \Qcircuit @C=0.7em @R=0.7em {
\lstick{\ket{0}}&\qw&\gate{U_m}&\ctrl{1}&\gate{U_m^{\dagger}}&\qw&\meter\\
\lstick{\ket{0}}&\qw&\gate{U_e}&\ctrl{1}&\gate{U_e^{\dagger}}&\qw&\meter\\
\lstick{\ket{\psi}}&\qw&\qw&\gate{-iX}&\qw&\qw&\qw\\}
\]
\vspace{1.58cm}
\caption{Circuit for multiplying mantissa rotation $U_m$ with exponent rotation $U_e$.  This circuit is a special case of that in ~\fig{gearbox} for the case where $d=2$.\label{fig:expcircuit}}
\end{minipage}

\end{figure}
}

\section{The Gearbox Circuit}\label{sec:gearbox}

The ``gearbox circuit'' is the central object that underlies our entire method.   The role of the circuit is to perform a rotation through an angle that is the product of the squares of the off--diagonal matrix elements of a series of single qubit unitary operations $U^1,\ldots,U^d$ acting on ancilla qubits.  We refer to this circuit as a gearbox circuit because it transforms coarse rotations into much finer rotations in analogy to a gearbox.  The circuit is denoted, in the case of $d$ control qubits, as $C^{(d)}(U^1,\ldots,U^d)$ and is given  in \fig{gearbox}.  The circuit is equivalent to those used in~\cite{CGM+09,CW12} to implement linear combinations of unitary operations in the case where one of the unitary operations is the identity.  

We use the circuit for three purposes: to multiply the rotation angles generated by $U_m$ and $U_e$ (see~\fig{expcircuit}), to reduce the spacing between the rotations that our circuits produce and to generate $U_e$.  
Our cost analysis assumes that Clifford operations ($H,S$ and CNOT) are inexpensive whereas the non--Clifford operation \T is expensive.  This cost model is motivated by the fact that {\T}gates
 are very expensive to perform  in many error correcting codes because multiple rounds of magic state distillation may be required to obtain sufficiently accurate \T gates~\cite{FSG09}.
The following theorem shows that the gearbox circuit can be used to convert modestly small rotations into very small rotations non--deterministically, and furthermore that the gearbox circuit can always be repeated until success is achieved because the rotation implemented when the circuit fails to give the desired rotation can be inverted using Clifford operations, which we assume are inexpensive.
\begin{theorem}
Given that each measurement in $C^{(d)}(U^1,\ldots,U^d)$ yields $0$, the circuit enacts the transformation $C^{(d)}(U^1,\ldots,U^d): \ket{0^{\otimes d}}\ket{\psi}\mapsto e^{-iX\tan^{-1}(\tan^2(\theta))}\ket{\psi}$, where $\sin^2(\theta)= |U^1_{1,0}|^2\cdots|U^d_{1,0}|^{2}$.  This outcome occurs with probability $\cos^4(\theta)+\sin^4(\theta)$ and all other measurement outcomes result in the transformation $\ket{\psi}\rightarrow e^{i\pi X/4}\ket{\psi}$, regardless of the choice of $U^1,\ldots,U^d$.\label{thm:smallrot}
\end{theorem}

\begin{proof}
There are three steps in the circuit, first $U_1$ through $U_d$ are performed on the ancilla qubits, then the $d$--controlled $-iX$ gate is applied and finally $U_1^\dagger$ through $U_d^\dagger$ are applied to the ancillas.  By applying these operators and expanding the matrix products that arise we find that
$C^{(d)}(U^1,\ldots, U^d)$ performs:
\begin{eqnarray}
{\ket{0^{\otimes d}}\ket{\psi}}&\rightarrow U^1\otimes\cdots\otimes U^d\ket{0^{\otimes d}}\ket{\psi}\nonumber\\
&\rightarrow U^1\otimes\cdots\otimes U^d\ket{0^{\otimes d}}\ket{\psi}-(U^1_{1,0}\cdots U^d_{1,0}\ket{1^{\otimes d}})(\openone+iX)\ket{\psi}\nonumber\\
&\rightarrow \ket{0^{\otimes d}}\ket{\psi}-\sum_j U^{1*}_{1,j_1}\cdots U^{d*}_{1,j_d} U^1_{1,0}\cdots U^d_{1,0}\ket{j}(\openone+iX)\ket{\psi}\nonumber\\
&= \left((1-|U^1_{1,0}|^2\cdots |U^d_{1,0}|^2) \openone -i|U^1_{1,0}|^2\cdots |U^d_{1,0}|^2X\right)\ket{0^{\otimes d}}\ket{\psi}\nonumber\\
&\qquad-\sum_{j\ne 0} U^{1*}_{1,j_1}\cdots U^{d*}_{1,j_d} U^1_{1,0}\cdots U^d_{1,0}\ket{j}(\openone+iX)\ket{\psi}\nonumber\\
&= \sqrt{\cos^4(\theta)+\sin^4(\theta)}\ket{0^{\otimes d}}\left(\frac{\cos^2(\theta) \openone -i\sin^2(\theta)X}{\sqrt{\cos^4(\theta)+\sin^4(\theta)}}\right)\ket{\psi}\nonumber\\
&\qquad-\sqrt{2}\sum_{j\ne 0} U^{1*}_{1,j_1}\cdots U^{d*}_{1,j_d} U^1_{1,0}\cdots U^d_{1,0}\ket{j}\left(\frac{\openone+iX}{\sqrt{2}}\right)\ket{\psi}.\label{eq:lem1proof}
\end{eqnarray}
It follows from trigonometry and the identity $e^{-i\phi X}=\cos(\phi)\openone -i\sin(\phi)X$ that~\eq{lem1proof} implies that the transformation $\ket{\psi}\rightarrow e^{-i\tan^{-1}(\tan^2(\theta))X}$ will be implemented by $C^{(d)}(U^1,\ldots,U^d)$ with probability $\cos^4(\theta)+\sin^4(\theta)$, and that the circuit implements the transformation $\ket{\psi}\rightarrow e^{i\pi X/4}\ket{\psi}$ in all other cases, as claimed.
\end{proof}

\begin{figure}[t]
 \includegraphics[width=0.6\textwidth]{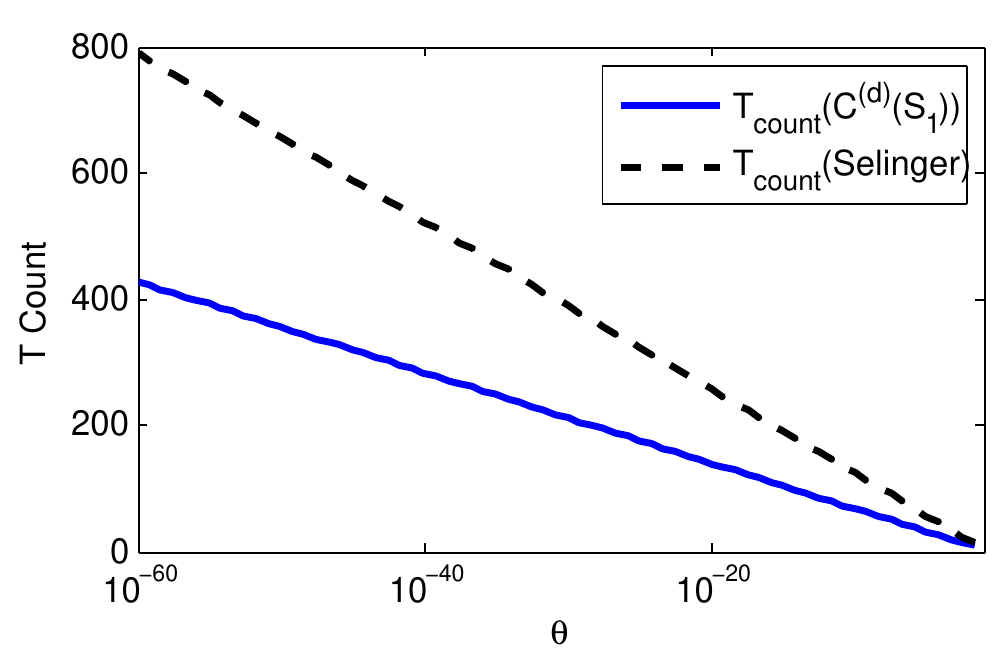}
\caption{Here we compare the mean \Tcount  for implementing $C^{(d)}(S_1)$, estimated using $500$ samples per angle and observe that gearbox circuits can be more efficient at synthesizing small rotations than Selinger's method.\label{fig:compare}}
\end{figure}

The \Tcount required to produce a rotation angle $\theta \approx |U_{1,0}^1|^2\cdots |U^d_{1,0}|^2$, given that each measurement outcome is $0$ and the simplified Tofolli circuit of~\cite{CJ13} is used to implement the $d$--controlled $-iX$ gate, is
\begin{equation}
T_{\rm count}(C^{(d)}(U^1,\ldots,U^d))= 4(d-1)+2\sum_{\ell=1}^d T_{\rm count}(U^\ell).\label{eq:tcount}
\end{equation}
Similarly, using the Toffoli construction of~\cite{CJ13} (depth two constructions that do not use measurement can be found in~\cite{AMM+12}) yields a \Tdepth of 
\begin{equation}
T_{\rm depth}(C^{(d)}(U^1,\ldots,U^d))= (d-1)+2\max_{\ell} T_{\rm depth}(U^\ell).\label{eq:tcount}
\end{equation}
Any failures that occur in implementing $C^{(d)}(S_\ell)$ can be corrected by applying Clifford operations and attempting the rotation again because $e^{iX\pi/4}$ is itself a Clifford operation, up to a global phase.  These estimates of the \Tcount and \Tdepth also approximately hold in cases where the rotation is attempted until success is obtained because \thm{smallrot} predicts that the failure probability will be very small if $\theta\ll 1$.  It is also interesting to note that similar circuits to our gearbox circuit have been proposed for implementing $V$--basis rotation~\cite{NC00}, suggesting that this template may be useful for a variety of tasks in quantum circuit synthesis.

\fig{compare} shows that the number of operations needed to synthesize a rotation of angle $\theta$ using $C^{(d)}(S_1)$  as a function of the rotation angle generated $\theta(d)$,
where in general $S_j$ is a unitary that yields the minimum value of $|{S_j}_{1,0}|$ over all $H$, \T circuits consisting of at most
$j$ $T$ gates and hence $S_1=HTH$.  The non--deterministic circuit manages to outperform a lower bound proven by Selinger~\cite{Sel12} for the number of \T gates needed to synthesize an arbitrary $Z$--rotation using the \ClT library and no ancilla qubits.
In fact, the \Tcount is smaller than that required for Selinger's method not just for very small rotations but also for the largest angles achievable using $C^{(d)}(S_1)$, which are on the order of $10^{-2}$ radians.
These results are significant because Selinger's circuit synthesis method is known to be optimal, meaning that there exist $Z$--rotations that require a number of \T gates that saturate the scaling predicted by the Selinger's method.
We will see in~\sec{optimal} that our non--deterministic circuits can in fact surpass the efficiency of any single qubit circuit synthesis method that uses our gate library and does not employ ancillary qubits.

It may be natural to suspect that the efficiency with which small angle rotations can be synthesized increases as $|U_{1,0}|$ decreases.  We find that using longer circuits to synthesize unitaries with smaller values of $|U_{1,0}|$ does not necessarily yield a more efficient method for generating small rotations.
\fig{slope} contains results found by fitting the \Tcount for $C^{(d)}(S_j)$ to a logarithmic function of the form $a(j)\log_2(\theta^{-1}) +b(j)$.  Using values of $j$ ranging from $1$ to $59$, we find strong evidence that $a(j)\approx 2$ is possible with this method.  This is superior to the method of Selinger, which gives $a(j)=4$ and we will show later in \fig{optslope} that this is also smaller than the optimal value of $a(j)\approx 3$  that arises using ancilla--free circuit synthesis using the gate library $\ClT\!\!\!$.  It should be noted, however, that $C^{(d)}(S_j)$ does not necessarily provide as fine control over the resultant rotation angle as these other circuit synthesis methods (especially for large $j$); although, our results show that it is more efficient at generating small rotation angles than the optimal ancilla free circuit synthesis method.

\begin{figure}[t]
 \includegraphics[width=0.8\textwidth]{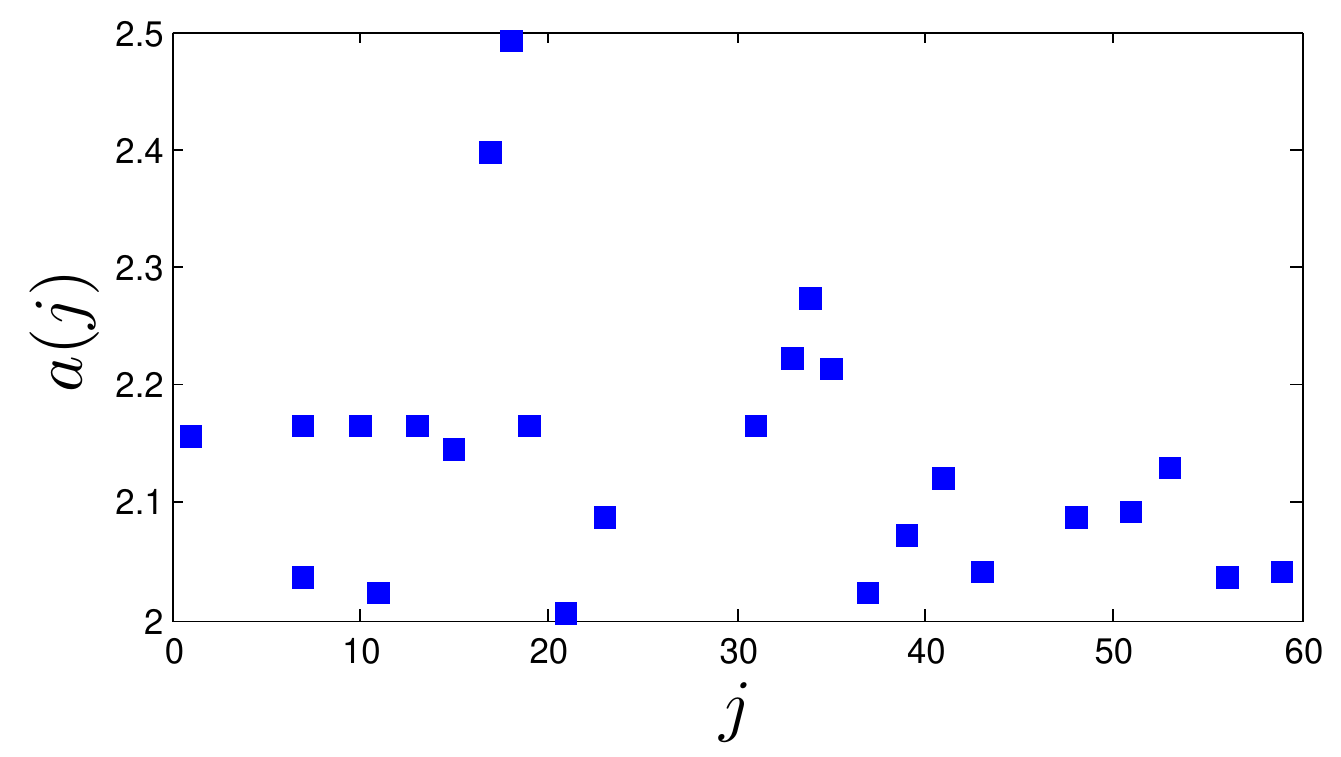}
\caption{Here we plot the fit parameter $a(j)$ for a least squares fit of the \Tcount as a function of $j$ for $C^{(d)}(S_j)$ to $a(j)\log_2(1/\theta(d))+b(j)$ for different values of $j$ where $\theta(d)$ is the rotation angle generated by $C^{(d)}(S_j)$ and $d$ is varied from $1$ to $128$ to estimate $a(j)$.  We see evidence from the data that the efficiency of generating a small angle rotation using $C^{(d)}(S_j)$ increases at first as a function of $j$ and then saturates. \label{fig:slope}}
\end{figure}

\section{The Composed Gearbox Circuit}\label{sec:compgearbox}

\fig{slope} shows that a direct application of the gearbox circuit requires a \Tcount that scales at least as $2\log_2(1/\theta)$, implying that a different approach is needed to further improve the scaling.
A natural way to improve on the prior method is to use the gearbox circuit recursively by taking $U$ to be the rotation yielded by another gearbox circuit.    This process can be repeated many times and the resulting circuit forms a tree--like structure as seen in \fig{recursivecircuit}.  We formally define the recursive construction of the ``composed gearbox circuit'' below.

\begin{definition}
Let $C^{\circ 1}(U)$ for $U\in U(2)$ be the circuit formed by taking $U^1=U$ in $C^{(1)}$, then for
for any integer $d>1$, $C^{\circ d}(U):=C^{\circ 1}(C^{\circ d-1}(U))$.\label{def:recursive}
\end{definition}

We then show in the following corollary that $C^{\circ d}(U)$ generates a rotation angle that scales as $\tan^{2^d}(\theta_0)$ in the limit of small $\theta$ (where $\sin^2(\theta_0)=|U_{1,0}|^2$).
\begin{corollary}\label{cor:successprob}
If each of the measurements in $C^{\circ d}(U)$ yield ``0'' then $C^{\circ d}: \ket{\psi}\rightarrow e^{-i\tan^{-1}(\tan^{2^{d}}(\theta_0))X}\ket{\psi}$ where $\sin(\theta_0)^2=|U_{1,0}|^2$.
\end{corollary}
\begin{proof}
We will first prove using induction that $C^{\circ d}(U)$ yields the transformation $e^{-i\tan^{-1}(\tan^{2^d}(\theta_0))X}$, given that the outcome of each measurement in the tree is $0$ and then use \thm{smallrot} to verify the claimed success probability.
The base case for our inductive proof, $C^{\circ 1}(U)$, has already been demonstrated by \thm{smallrot} for the case where $d=0$.  Now let us assume that $C^{\circ d-1}(U)$ enacts $e^{-i\tan^{-1}(\tan^{2^{d-1}}(\theta_0))X}$.  The off--diagonal matrix elements of this matrix have magnitude $|\sin(\tan^{-1}(\tan^{2^{d-1}}(\theta_0)))|$ and hence it follows from \thm{smallrot} that $C^{\circ 1}(C^{\circ d-1}(U))$ enacts, upon success,
\begin{equation}
e^{-i\tan^{-1}(\tan^{2\cdot 2^{d-1}}(\theta_0))X}=e^{-i\tan^{-1}(\tan^{2^{d}}(\theta_0))X},
\end{equation}
as claimed.  
\end{proof}

\begin{figure}[t!]
\label{fig:recursivecircuit}
\[
\hspace{-2cm}
\tiny
\newcommand{\up}[1]{\push{\raisebox{6pt}{$#1$}}}
\Qcircuit @C=0.4em @R=0.4em {
\lstick{\ket{0}}&\qw&\gate{U}&\ctrl{1}&\gate{U^\dagger}&\qw&\meter&&\lstick{\ket{0}}&\qw&\gate{U}&\ctrl{1}&\gate{U^\dagger}&\qw&\meter&&\lstick{\ket{0}}&\qw&\gate{U}&\ctrl{1}&\gate{U^\dagger}&\qw&\meter&&&\lstick{\ket{0}}&\qw&\gate{U}&\ctrl{1}&\gate{{U}^\dagger}&\qw&\meter&&&\\
\lstick{\ket{0}}&\qw&\qw&\gate{-iX}&\qw&\qw&\qw&\ctrl{1}&\qw&\qw&\qw&\gate{iX}&\qw&\qw&\meter&&\lstick{\ket{0}}&\qw&\qw&\gate{-iX}&\qw&\qw&\qw&\ctrl{1}&\qw&\qw&\qw&\qw&\gate{iX}&\qw&\qw&\meter\\
\lstick{\ket{0}}&\qw&\qw&\qw&\qw&\qw&\qw&\gate{-iX}&\qw&\qw&\qw&\qw&\qw&\qw&\qw&\qw&\ctrl{1}&\qw&\qw&\qw&\qw&\qw&\qw&\gate{iX}&\qw&\qw&\qw&\qw&\qw&\qw&\qw&\meter\\
\lstick{\ket{\psi}}&\qw&\qw&\qw&\qw&\qw&\qw&\qw&\qw&\qw&\qw&\qw&\qw&\qw&\qw&\qw&\gate{-iX}&\qw&\qw&\qw&\qw&\qw&\qw&\qw&\qw&\qw&\qw&\qw&\qw&\qw&\qw&\qw&\qw&\\
}
\]
\caption{A circuit expansion of $C^{\circ 3}(U)$.  Note that every non--Clifford operation except the right most $U^\dagger$ can be implemented using ancillas containing $C^{\circ 2}(U)\ket{0}$, $C^{\circ 1}(U)\ket{0}$ and $U\ket{0}$.  \label{fig:recursivecircuit}}
\end{figure}
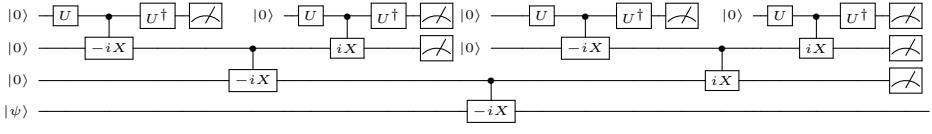

One of the most remarkable features of $C^{\circ d}(U)$ is that almost all of the computational steps in the circuit can be thought of as preparations of ancilla states either of the form $\ket{\theta_j}:=C^{\circ j}(U)\ket{0}$ for $j=1,\ldots,d-1$ or $U\ket{0}$.  In fact, all but $1$ application of $U^\dagger$ can be implemented as ancilla preparations that are performed offline.  This means that the ancilla preparations can be performed prior to attempting the rotation, potentially by using multiple quantum information processors working in parallel.  In contrast, the final application of $U^\dagger$ cannot be performed in this manner and hence  is an online cost.  We do not discuss the success probability in \cor{successprob} because it varies depending on whether ancillas containing $\ket{\theta_j}$ are provided or not.  We show below that if such ancillas are provided then the success probability is bounded below by a constant for all $d$.  In constrast, we will see
that if no ancillas
are provided then, with high probability, multiple rounds of error correction will be needed for the algorithm to succeed with high probability.



\begin{lemma}
For all integer $d>0$ and $\theta_0 < \pi/4$, if ancilla qubits of the form $U\ket{0}$ and $\ket{\theta_j}:=C^{\circ j}(U)\ket{0}$ for $j=1,\ldots,d-1$ are provided then $C^{\circ d}(U)$ can be implemented with failure probability at most $$P_{\rm fail}<\frac{1-\cos(4\theta_0)}{4}+\frac{2\tan^4(\theta_0)}{1-\tan^2(\theta_0)}.$$\label{lem:ancillaprob} 
\end{lemma}
\begin{proof}
We know from \thm{smallrot} that the probability of successfully implementing $C^{\circ 1}(U)$ is $\cos(\theta_0)^4+\sin(\theta_0)^4$.  \cor{successprob} similarly tells us that the probability that the $j^{\rm th}$ measurement is successful given that $j\ge 2$ and all prior measurements were successful is
\begin{eqnarray}
 P_{\rm success}(j|j-1,\ldots, 1)&=\cos^4(\tan^{-1}(\tan^{2^j}(\theta_0)))+\sin^4(\tan^{-1}(\tan^{2^j}(\theta_0)))\nonumber\\
&=\frac{1+\tan^{2^{j+2}}(\theta_0)}{(1+\tan^{2^{j+1}}(\theta_0))^2}.
\end{eqnarray}
Therefore the probability of failure at step $j$, given success at all previous steps, obeys
\begin{equation}
P_{\rm fail}(j|j-1,\ldots,1)=\frac{2\tan^{2^{j+1}}(\theta_0)}{(1+\tan^{2^{j+1}}(\theta_0))^2}\le 2\tan^{2^{j+1}}(\theta_0).
\end{equation}
The probability of a failure occuring is at most the sum of the probabilities of failing at any given step and hence
\begin{eqnarray}
P_{\rm fail}&\le 1- \cos(\theta_0)^4-\sin(\theta_0)^4 +\sum_{q=1}^{d-1}2\tan^{2^{q+1}}(\theta_0)\nonumber\\
&\le 1- \cos(\theta_0)^4-\sin(\theta_0)^4 + \sum_{q=1}^{\infty}2\tan^{2(q+1)}(\theta_0)\nonumber\\
&\le \frac{1-\cos(4\theta_0)}{4}+\frac{2\tan^4(\theta_0)}{1-\tan^2(\theta_0)}.\label{eq:pfail}
\end{eqnarray}

\end{proof}
The upper bound on the success probability given by \lem{ancillaprob} can be used to estimate the number of times the circuit needs to be attempted, in cases where ancillas are provided since assuming the presence of ancilla states that contain $\ket{\theta_j}$ for $j=1,\ldots,d-1$ is equivalent to assuming that all previous computational steps have already been successfully implemented.  We expand on this reasoning in the following corollary.
\begin{corollary}\label{cor:expect}
For integer $d>0$ and $\theta_0< \pi/4$, the number of ancilla states of each type and the number $U^\dagger$ operations, $N_d$, that must be performed online to execute the circuit $C^{\circ 1}(U)$ successfully follows a probability distribution with mean and variance obeying
\begin{eqnarray}
\mathbb{E}(N_d) &\le \left(\frac{3+\cos(4\theta_0)}{4}+\frac{2\tan^4(\theta_0)}{1-\tan^2(\theta_0)}\right)^{-1},\nonumber\\
\mathbb{V}(N_d) &\le \frac{\frac{1-\cos(4\theta_0)}{4}-\frac{2\tan^4(\theta_0)}{1-\tan^2(\theta_0)}}{\left(\frac{3+\cos(4\theta_0)}{4}+\frac{2\tan^4(\theta_0)}{1-\tan^2(\theta_0)}\right)^2}.\label{eq:corcost}
\end{eqnarray}
\end{corollary}
\begin{proof}
The number of times the measurement has to be repeated, $N_d$, is geometrically distributed with mean $1/P_d$ and variance $(1-P_d)/P_d^2$, where $P_d $ is the probability of the measurement succeeding.  Since the mean and the variance are monotonically increasing functions of $P_{\rm fail}$ therefore upper bounds for $\mathbb{E}(N_d)$ and $\mathbb{V}(N_d)$ can be found by substituting~\eq{pfail} into them because at most one of each of these types of resources are needed to attempt to implement $C^{\circ{d}}(U)$.  The proof of the corollary then follows by simplifying the result of this substitution.
\end{proof}

As an example, we find from substituting $\theta_0=\pi/8$ into~\eq{corcost} that the number of trials needed to implement $C^{\circ d}(HTH)$ follows a distribution with
$\mathbb{E}(N_d)< \frac{5}{4}$ and $\mathbb{V}(N_d)<\frac{1}{3}$.
Chebyshev's inequality then implies that if we define $X$ to be the number of trials needed to achieve a successful rotation then
\begin{equation}
{\rm Pr}(|X-\mathbb{E}(N_d)|\ge \chi)< \frac{1}{3\chi^2}.\label{eq:probbd}
\end{equation}
This implies that with high probability the number of each type of resource consumed in implementing the successful rotation is  a constant.  If the cost of each of these resources is assumed to be identical, then the cost of the algorithm is $O(d)=O(\log\log(\theta^{-1}))$ and the online cost of implementing the circuit is bounded above by a constant, with high probability.

The mean and the variance of the number of $U$ and $U^\dagger$ operations used to implement the rotation can also be computed in cases where no precomputed ancillas are provided.  In fact, the number of $U$ and $U^\dagger$ gates that are
needed to implement $C^{\circ d}(U)$ with high probability scales as $O(2^d)$.  We state this result in the following theorem.

\begin{theorem}\label{thm:fullcost}
Let $P_q=\sin(\phi_q)^4+\cos(\phi_q)^4$, where $\phi_q:= \tan^{-1}(\tan^{2^{q-1}}(\theta_0))$ for all integer $q\ge 1$ and let $n_d$ be a random variable representing the number of applications of $U$ or $U^\dagger$ used to enact $C^{\circ d}(U)$ in a given attempt.  Then the expectation value of $n_d$ is
$$\mathbb{E}(n_d) = \frac{2^d}{P_1\cdots P_{d}},$$
and for $\theta_0< \pi/4$ the variance of $n_d$ obeys
$$\mathbb{V}(n_d)\le \frac{2^{2d+1}(1-P_1)}{P_d^2\cdots P_1^2}\left(1+\frac{(P_d\cdots P_1)}{2^{d}P_1}\right).$$
\end{theorem}
To prove \thm{fullcost}, we think about our non-determinitic circuits  as ones that always succeed, but require
a random number of steps to do so. We introduce two random variables to describe the number of measurements required for the measurement at the $n^{\rm th}$ level of our tree to succeed: one that describes number of attempts needed
to successfully execute the branch before the  controlled $-iX$ at the $n^{\rm th}$ level and the other describes the number of attempts needed for the branch after the controlled $-iX$ and before all measurements. We then express the mean and variance of the number of attempts required to execute the $n^{\rm th}$ level of
the tree in terms of the mean and variance of the variables introduced to describe the number of attempts needed
to succeed on the $(n-1)^{\rm st}$ level.  We get a recursive relation for mean and variance that we then unfold and simplify using simple upper bounds. The same idea can be used
to analyze more complicated tree-like non-deterministic circuits.
Proof is given in~\ref{app:fullcost}.

\thm{fullcost} shows that the mean and the standard deviation of the number of applications of $U$ and $U^\dagger$ used to implement $C^{\circ d}(U)$ scales as $\Theta(2^d)$ and $O(2^d)$ respectively for $\theta_0\le \pi/8$.
This follows from the fact that for $\theta_0\le \pi/8$,
\begin{eqnarray}
\frac{1}{P_d\cdots P_1}&\le \frac{1}{(1-2\tan^{2^{d}}(\theta_0))\cdots (1-2\tan^2(\theta_0))}\le\frac{1}{\exp(-4\sum_{k=1}^d\tan^{2^{k}}(\theta_0))}\nonumber\\
&\le\frac{1}{\exp(-4\sum_{k=1}^d\tan^{2{k}}(\theta_0))}\nonumber\\
&\le \exp\left(\frac{4\tan^{2}(\theta_0)}{1-\tan^{2}(\theta_0)}\right)< \frac{5}{2}.
\end{eqnarray}

Chebyshev's inequality therefore implies (similarly to the case discussed above where precomputed ancillas are used) that, with high probability, the number of $U$ and $U^\dagger$ gates needed to implement the rotation will also scale as $O(2^d)$.  This procedure  is efficient because $d$ scales doubly--logarithmically with the desired rotation angle.  The complexity of implementing $C^{\circ d}(U)$ is therefore logarithmic in $1/\theta$ for any fixed $U$ with $\theta_0\le \pi/8$.

\begin{figure}[t!]
 \includegraphics[width=0.7\textwidth]{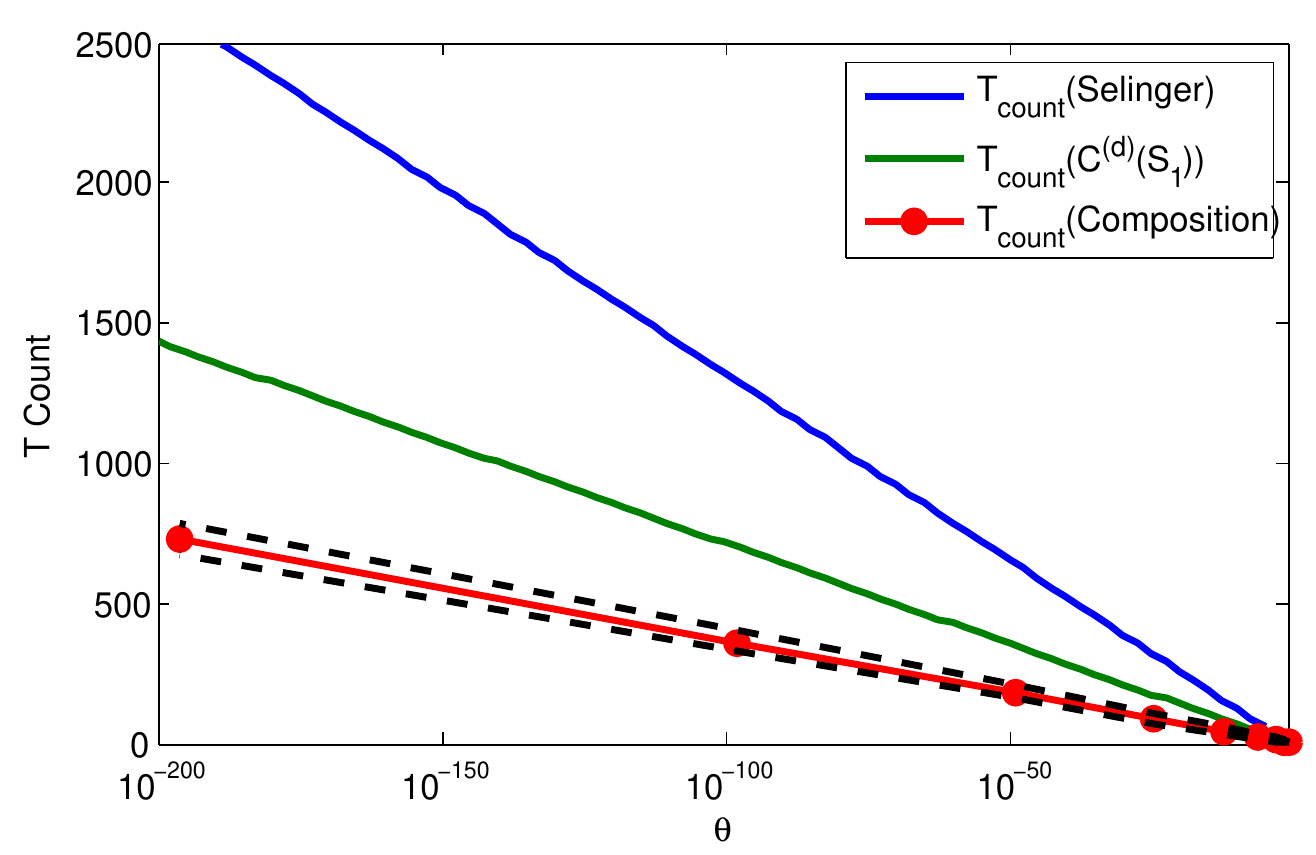}
\caption{Here we compare the mean \Tcount for our composition based method given by $C^{\circ d}(HTH)$ to Selinger's method and also directly using the gearbox circuit $C^{(d)}(S_1)$.  The dashed lines give the upper and lower limits of a $95\%$ confidence interval for the \Tcount that arises from using the composition method to $U_e$.  We see that the composition method offers superior performance to that of the circuit $C^{(d)}(S_j)$ and Selinger's method.  $500$ samples were used to compute the expectation values of the scalings for both non--deterministic methods.    \label{fig:compare2}}
\end{figure}


This implies that, on average, the number of \T gates required to implement $e^{-i \tan^{-1}(\tan^{2^d}(\pi/8)) X}$ is at most
\begin{equation}
\mathbb{E}(n_d)\le 5\cdot 2^{d-1} < 2\log_2(1/\tan(\theta)),
\end{equation}
where $\theta=\tan^{-1}(\tan^{2^d}(\pi/8))$.   This estimate results from the use of several inequalities and it is therefore reasonable expect the actual expectation value of the \T  count to be smaller.
  The data in~\fig{compare2} suggest that the mean value for the \T  count (which is proportional to $n_d$ for $U=HTH$) actually obeys
\begin{equation}
\mathbb{E}(n_{d(\theta)})\approx 1.11 \log_2(\theta^{-1})-0.01,\label{eq:compscale}
\end{equation}
for $d\in \Theta (\log(\log(\theta^{-1})))$. The resultant \Tcount is smaller than that of~\cite{Sel12}  (which is known to give optimal scaling in cases where $\theta$ is chosen adversarially and no ancilla bits are permitted)  or those that arise from a direct application of the gearbox circuit.
We will see shortly that this scaling is in fact better than the best possible scaling achievable in any circuit synthesis method using only $H,T$ and CNOT gates.
Furthermore, the slopes of the $2.5^{\rm th}$ and $97.5^{\rm th}$ percentile of the \Tcount are approximately $1.04$ and $1.18$ respectively.  We extend the $x$--axis to $10^{-200}$ radians (which is unreasonably smal for most applications) to accurately assess the scaling and emphasize that relatively small values of $d$ can lead to miniscule rotation angles.  This suggests that small rotations generated by $C^{\circ d}(HTH)$ will have, with high probability, smaller \Tcounts than existing methods.
A drawback of using $C^{\circ d}(HTH)$ as opposed to $C^{(d)}(HTH)$ to generate $U_e$ is that $C^{\circ{d}}$ generates small rotation angles that scale as $\tan^{2^d}(\pi/8)$, which does not give fine control over the rotation angle if only the variable $d$ is used to control
the rotation.  

\begin{figure}[t]
 \includegraphics[width=0.8\textwidth]{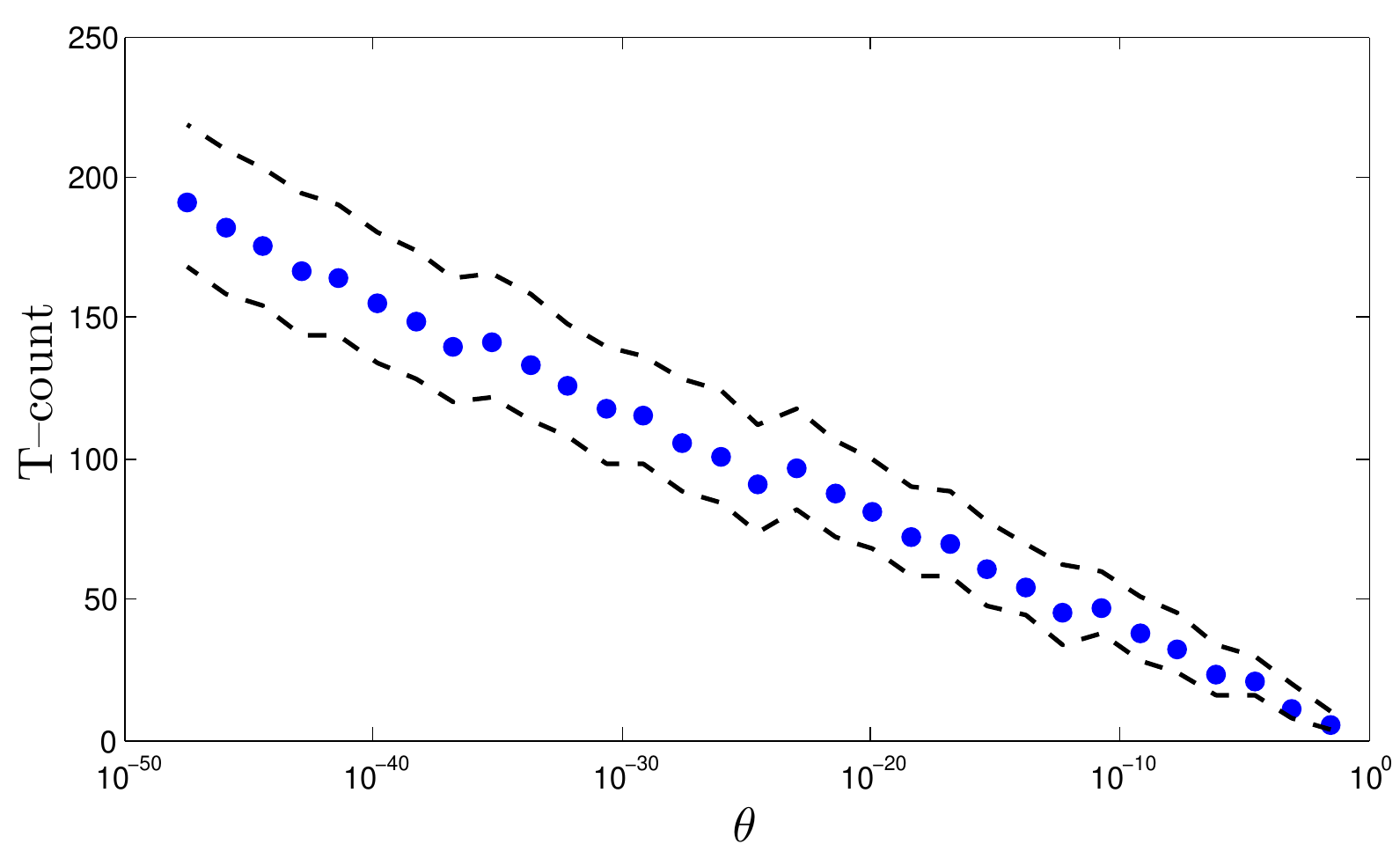}
\caption{Here we plot the mean \Tcount for $C^{d}(C^{\circ j_1}(HTH),\ldots,C^{\circ j_d}(HTH))$ as a function of the rotation angle generated by the circuit.  The dashed lines give the upper and lower limits of a $95\%$ confidence interval for the \Tchar--count, and the dots show the average \Tchar--count.  The data scales approximately as $a\log_2(1/\theta)+4.2$, where the value of $a$  that gives the least--square error is $1.14$ and $a\in [1.05, 1.20]$ with probability $0.95$.  2000 samples were used to find the distribution of the \Tcount for each value of $\theta$.  \label{fig:skexponent}}
\end{figure}


The problem of poor control over the rotation angle used for $U_e$ can be addressed, at a modest cost, by using the gearbox circuit $C^{(d)}$ in tandem with the composed gearbox circuit $C^{\circ d}(HTH)$.  In particular, let $D_1,\ldots,D_d$ be positive integers.  Then $C^{(d)}(C^{\circ D_1}(HTH),\ldots,C^{\circ D_d}(HTH))$ non-deterministically implements $e^{-i\phi X}$ for
\begin{eqnarray}
\phi=\tan^{-1}[\tan^{2}\big(\phi(D)\big)]\approx (0.1716)^{2^{D_1}+\cdots+2^{D_d}},\label{eq:phiDapprox}
\end{eqnarray}
where
\begin{equation}
\phi(D):=\sin^{-1}\left[\sin\left(\tan^{-1}\big(\tan^{ 2^{D_1}}(\pi/8)\big)\right)\times\cdots\times\sin\left(\tan^{-1}\big(\tan^{ 2^{D_d}}(\pi/8)\big)\right)\right].\label{eq:phiddef}
\end{equation}
By using a binary expansion and a Taylor series expansion of the trigonometric functions, it can be seen that the circuit implements $e^{-i\phi X}$ for $\phi=\tan^{4q}(\pi/8)+O(\tan^{12q}(\pi/8))$ and integer $q$.  This allows us to address the problems posed by using our composition method to construct the rotation angle at the cost of additional \T gates.

\fig{skexponent} contains a plot of the rotation angles generated by combining the rotations generated using our composition method via the gearbox circuit.   We see in the figure that the rotation angles obtained approximately decrease by factors $0.031$, as anticipated by the prior discussion.
We also find that the expectation value of the \Tcount of this algorithm scales roughly as $a\log_2(1/\theta)+4.2$ where  $a\approx 1.14$ giving the line of best fit and $[1.05,1.20]$ gives a $95\%$ confidence interval for $a$.   The typical overhead from using $C^{(d)}(C^{\circ D_1},\ldots,C^{\circ D_d})$ to implement the rotation is minimal because the cost of implementing a small rotation using $C^{\circ d}(HTH)$ followed a similar scaling with $a\approx 1.11$, which falls within the $95\%$ confidence interval for the value of $a$ corresponding to $C^{(d)}(C^{\circ D_1},\ldots,C^{\circ D_d})$.

\section{Constructing the Floating Point Representation}\label{sec:floatingpoint}
The preceding discussion shows how we can use our composition method in conjunction with the gearbox circuit to implement a given $U_e$.  Our next goal is to use this idea to implement an arbitrary $X$--rotation by using this method to generate the exponent of our floating point representation, $U_e$, and another technique to implement the mantissa $U_m$.  The circuit that implements the necessary rotation is given in \fig{corfloat}.

\thm{smallrot} implies that, conditioned on the successful implementation of the $C^{\circ D_j}(HTH)$, the circuit will implement $e^{-i\phi X}$ for
\begin{equation}
\phi=\tan^{-1}\Big(\tan^{2}\left[\sin^{-1}\large(|U_{m_{1,0}}|\sin\big(\phi(D)\big)\large)\right]\Big)\approx |U_{m_{1,0}}|^2 \tan^{2(2^{D_1}+\cdots 2^{D_d})}(\pi/8),
\end{equation}
where $\phi(D)$ is defined in~\eq{phiddef}.

We describe the process involved in using this floating point implementation of the rotation below.

\begin{algorithm}
\caption{Floating Point Implementation of $e^{-i\phi_{\rm in} X}$.}
\label{alg:1}
\begin{algorithmic}[1]
\Require $\alpha$, $\gamma$ such that $\alpha\times 10^{-\gamma}:=\phi_{\rm in}$, $0<\alpha<1$ and $\gamma$ is an integer, $\delta>0$, Quantum state $\ket{\psi}$, A circuit synthesis algorithm $\mathcal{C}: (U(2),\mathbb{R})\mapsto U(2)$ such that for all $U\in U(2)$ and $\epsilon\ge 0$, $\|\mathcal{C}(U,\epsilon)-U\|\le \epsilon$.
\Ensure A quantum state approximating $e^{-i\phi X}\ket{\psi}$ within error $O(\delta \times 10^{-\gamma})$.

\State Set $k= \lfloor \phi_{\rm in}/ (\pi/4) \rfloor$.\label{algstep:start}
\State $\ket{\psi}\rightarrow HS^kH\ket{\psi}$.
\If{$|\phi_{\rm in}- k\pi/4|\le \delta$}
\State \Return $\ket{\psi}$\label{algstep:returnpsi}
\Else\label{algstep:else}
\State Set $\phi_{\rm rem}=\phi_{\rm in} - k\pi/4$.
\State Find the smallest value of $\phi(D)$, and the corresponding values of $D_1<\cdots<D_d$, such that $$\sin(\phi(D))\ge \sqrt{\frac{\tan \phi_{\rm rem}}{1+\tan \phi_{\rm rem}}}.$$\label{algstep:1}
\State Set $U_m=\mathcal{C}(\exp(-i\tilde{\phi} X),\delta)$ where $$\tilde{\phi}=\sin^{-1}\left(\frac{1}{\sin(\phi_D)}\sqrt{\frac{\sin \phi_{\rm rem}}{\cos \phi_{\rm rem} +\sin \phi_{\rm rem}}}\right).$$\label{algstep:final}
\State \Return $C^{(d+1)}(\mathcal{C}(U_m,\delta),C^{\circ D_1}(HTH),\ldots,C^{\circ D_d}(HTH))\ket{\psi}$.
\EndIf

\end{algorithmic}
\end{algorithm}

The algorithm can be seen to output the desired rotation via the following argument.
It is easy to see that steps~\ref{algstep:start}--\ref{algstep:returnpsi} will return a distance $\delta$ approximation to
the rotation angle, given that the desired rotation obeys $\min_{k\in \mathbb{Z}}|\phi_{\rm in} -k\pi/4|\le \delta$.
The remaining cases can then be handled by implementing  $e^{-ikX\pi/4}$ using Clifford operations and synthesizing a rotation that implements $e^{-i(\phi_{\rm in} -k\pi/4)X}$ within precision $O(\delta \times 10^{-\gamma})$.

We have from \thm{smallrot} that the rotation angle implemented, for the ideal choice of $U_{m}$ is
\begin{equation}
\phi_{\rm rem}=\tan^{-1}(\tan^2(\sin^{-1}(|(U_{m})_{1,0}| \sin \phi(D))))\!=\! \tan^{-1}\!\left(\!\frac{|(U_{m})_{1,0}|^2 \sin^2 \phi(D)}{1-|(U_{m})_{1,0}|^2 \sin^2 \phi(D)}\!\right).\label{eq:finrot}
\end{equation}
We are constrained, however, to have $|(U_{m})_{1,0}|\le 1$ in our solution.  We find the range of physically allowable solutions by setting $|(U_{m})_{1,0}|=1$ and then solving for $\phi(D)$ to find that
a valid solution exists if \begin{equation}
\sin(\phi(D))\ge \sqrt{\frac{\tan \phi_{\rm rem}}{1+\tan \phi_{\rm rem}}},
\end{equation}
which is guaranteed by Step~\ref{algstep:1}.
Then given any such choice of $D$,  we  solve~\eq{finrot} for the corresponding value of $|(U_m)_{1,0}|$ and  find that
\begin{equation}
|(U_{m})_{1,0}| = \frac{1}{\sin(\phi_D)}\sqrt{\frac{\sin \phi_{\rm rem}}{\cos \phi_{\rm rem} +\sin \phi_{\rm rem}}}\in \Theta(1).
\end{equation}
The $x$--rotation chosen in Step~\ref{algstep:final} yields the desired rotation and hence the algorithm will as well, modulo the error incurred in the synthesis of $U_m$.

\begin{figure}[t!]
\[
\small
\newcommand{\up}[1]{\push{\raisebox{6pt}{$#1$}}}
 \Qcircuit @C=0.7em @R=0.7em {
\lstick{\ket{0}}&\qw&\gate{U_m}&\ctrl{1}&\gate{U_m^{\dagger}}&\qw&\meter\\
\lstick{\ket{0}}&\qw&\gate{C^{\circ D_1}(HTH)}&\ctrl{1}&\gate{C^{\circ D_1{\dagger}}(HTH)}&\qw&\meter\\
\lstick{\ket{0}}&\qw&\gate{C^{\circ D_2}(HTH)}&\ctrl{1}&\gate{C^{\circ D_2 {\dagger}}(HTH)}&\qw&\meter\\
&&\up{\vdots}&\up{\vdots}&\up{\vdots}&&\\
\lstick{\ket{0}}&\qw&\gate{C^{\circ D_d}(HTH)}&\ctrl{-1} \qwx[1]&\gate{C^{\circ D_d {\dagger}}(HTH)}&\qw&\meter\\
\lstick{\ket{\psi}}&\qw&\qw&\gate{-iX}&\qw&\qw&\qw\\
}
\]
\caption{This circuit gives the floating point implementation of a rotation for a given mantissa unitary $U_m$.  Unlike~\fig{expcircuit}, this circuit uses $d$ different composed gearbox circuits
to form the exponent part rather than just one.  This provides greater control over the rotation than would be possible with just one composed gearbox.   Note that that the multiply contrilled $-iX$ gate can be implemented using $4(d-1)$ \T gates as discussed in~\cite{CJ13}.  \label{fig:corfloat}}
\end{figure}
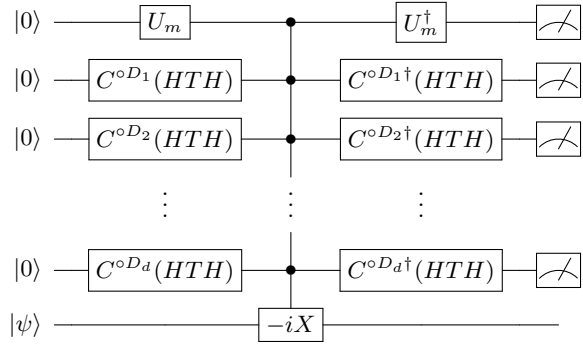
We have already established in~\eq{phiDapprox} that $\phi_{D}$ will be within a constant factor of $\phi_{\rm rem}$, and hence $\phi_D \in \Theta(10^{-\gamma/2})$.  We then see from Taylor's theorem that
\begin{eqnarray}
\tan^{-1}\left(\frac{(|(U_{m})_{1,0}|+\delta)^2 \sin^2 \phi(D)}{1-(|(U_{m})_{1,0}|+\delta)^2 \sin^2 \phi(D)}\right)\nonumber\\
 \qquad=  \tan^{-1}\left(\frac{|(U_{m})_{1,0}|^2 \sin^2 \phi(D)}{1-|(U_{m})_{1,0}|^2 \sin^2 \phi(D)}\right) + O(\delta \sin^2 \phi(D)),
\end{eqnarray}
which verifies that the error is $O(\delta \times 10^{-\gamma})$ as required.

A cost analysis of the floating point method is given in \app{cost}, wherein we show that the $T$--count required by the floating point method approximately scales as $1.14 \log_2(1/\theta)$ for constant precision.  Similarly, the circuit depth and the online $T$--count scale as $O(\log\log(1/\theta))$.  This implies that floating point synthesis is not only less expensive than traditional synthesis methods (as measured by the $T$--count) but much of this cost can be distributed over parallel quantum information processors.

\section{Example: Implementing $\mathbf{\textbf{exp}(-i \pi Z/2^{16}})$ Using Floating Point Synthesis\label{sec:example}}
We will now give an illustrative example of our floating point technique for synthesizing the operation $e^{-i\pi/2^{16}Z}$.  This rotation is significant because it appears in the quantum Fourier transform.  We have found, by using techniques described in the subsequent section and~\cite{prappr}, that the \T--optimal circuit that estimates this rotation more accurately than $e^{-i\pi/2^{16}Z}\approx \openone$ consists of $57$ \T gates.  The next shortest circuit contains $60$ \T--gates.  This implies that the cost of synthesizing the rotation using an optimal circuit synthesis method and the $\ClT$ gate library changes abruptly when an approximation to the rotation with even one digit of precision is needed.

First, note that $R_z(\theta)=HR_x(\theta)H$ and hence the $x$ rotations that naturally arise from our method can be easily translated to $z$--rotations using Clifford operations (which we assume are inexpensive).  This implies that the problem of synthesizing the rotation reduces to that of synthesizing $e^{-i \pi/2^{16} X}$.  Following Algorithm~\ref{alg:1}, we choose $U_e$ to be $C^{(2)}(\pi/8)$ because $\tan^{-1}(\tan^4(\pi/8))> \sqrt{\pi/2^{16}}$.  We then find numerically that the mantissa part of the rotation must satisfy $$|({U_{m}})_{1,0}|\approx 0.235.$$  Finally, we exhaustively search for the two shortest circuits that give a unitary that has off--diagonal matrix elements of comparable magnitude to the ideal value and examine the performance of our floating point method for both these choices of $U_m$ by performing a Monte--Carlo simulation of the \Tcounts required to use our floating point method.
The results of this Monte--Carlo simulation are given in \tab{1}.


We see from the data in \tab{1} that circuits derived from the floating point method require, with high probability, nearly half the \T gates required by the optimal synthesis method in order to produce non--trivial approximations with comparable relative error.  As the desired relative error shrinks, floating point synthesis begins to lose its advantage the cost of the mantissa circuit will eventally approach half the cost of synthesizing the rotation. Thisl results in an  approximation that is inferior to optimal single--qubit systnehsis because the mantissa circuit must be applied twice. We see that in the case where a mantissa circuit with $29$ $T$ gates is used, requires a comparable number of $T$ gates to the optimal single qubit rotation $S_7$ but incurres nearly $5$ times the error.  We discuss the regime where floating point synthesis yields a superior \Tcount to optimal single qubit synthesis in detail in~\app{cost}.

It is easy to also see that larger rotations can also benefit from floating point synthesis.  For example, consider $\exp(-i \pi Z/ 2^8)\approx \exp(-i 0.0123 Z)$.  In this case, we see from~\tab{1} that synthesizing this rotation within $1$ digit of precision requires a minimum of $11$ $T$--gates using the single qubit Clifford, $T$ gate library.  In contrast, floating point synthesis can achieve the same rotation using on average $9.2$ $T$ gates (using $U_m=H$ and $U_e = C^{(1)}(HTH)$).  This shows that the floating point synthesis can be valuable for synthesizing even modestly large rotations.

The floating point circuits also have the benefit of requiring a substantially smaller online cost (meaning that many of the required operations can be implemented using precomputed ancillas~\cite{CJ12}).  For the cases considered in \tab{1}, these costs are approximately $8,11$ and $34$ \T gates and the majority of the online cost is incurred in implementing the Toffoli gate and $U_m^\dagger$ ($U_m$ can be implemented offline).  The circuits also are more resilient to gate faults and approximate the rotation with an axial rotation (in contrast to conventional methods).  Such costs could be further reduced by using variants of gearbox circuits to synthesize $U_m$.  For these reasons, floating point synthesis can provide more desirable circuits than traditional synthesis methods even if it does not lead to a substantial reduction in the $T$--count.

\begin{table}[t!]
\caption{This table compares the \Tcounts that result from synthesizing $e^{-i Z \pi/2^{16}}$ using our floating point method to those that arise from optimal synthesis using the gate library $\ClT\!\!\!$.   $V_1$ and $V_2$ are the two shortest circuits that provide a better approximation to the rotation than $e^{-i Z \pi/2^{16}}\approx \openone$.  $M_{29}$ is the most accurate approximation to $U_m$ possible using $29$ $T$--gates (without using ancillas).  The mean and confidence intervals were calculated using $40000$ samples and the mean value agrees with the result of \thm{fullcost} within statistical error.\label{tab:1}}
 
 {\footnotesize
\begin{tabular}{|c|c|c|c|l|}
\hline
$U_m$ & Mean      & Variance & $95\%$ Confidence  & Relative\\
      & \Tcount &          & Interval & Error\\
\hline
$HZ THZ THZ TH$ & 21.3 & 11.0& [18,30] & 0.35\\
$HTHTHTHTHTHTH$ & 27.3&11.0 & [24,36] & 0.13\\
$M_{29}$ &73.3 &11.0 &[70,82] & 0.0029\\
\hline
\hline
Circuit & \Tcount & -- & --& Relative \\
 &  &  &   & Error\\
\hline
$V_1$ &  57 & & &0.17\\
$V_2$ &  60 & & & 0.058\\
$V_7$ & 71 & & & 0.00056\\
\hline
\end{tabular}
}
\end{table}

\section{Optimal Ancilla--Free Single--Qubit Synthesis of Small Rotations}\label{sec:optimal}
In this section, we extend methods described in \cite{prappr} to find circuits chosen from the $\ClT$ library with the smallest possible (non-zero) off-diagonal entries. The algorithm described guarantees optimality of the found circuits. The result of the section shows that gear box circuits involving ancillary qubits and measurement reduces the \Tcounts below the best possible \Tcounts in a purely unitary single qubit construction. This shows that the use of ancillas and measurement leads to a significant advantage for synthesizing rotations.

More precisely, the problem we are interested in is the following: amongst all circuits with optimal \Tcount $n$ find one that corresponds to a unitary with a minimal possible off-diagonal entry. We say that circuit has optimal \Tcount $n$ if any other circuit drawn from $\ClT$~library implementing the same unitary requires at least $n$ \T gates. We reduce the problem to searching for unitaries over the ring 
\[
\Zr:= \left\{ \left. \frac{a+b\w+c\w^2+d\w^3}{\sqrt{2}^\K} \right| a,b,c,d,\K \in \Z \right\}, \w:=e^{i \pi/4}
\]
with a certain property that we discuss in detail later in this section. It is known that any circuit over \ClT library corresponds to a unitary over $\Zr$; furthermore, the results presented in \cite{es} show that there is a tight connection between optimal \Tcount and entries of the unitary. The notion of the smallest denominator exponent ($\sde$) allows us to express the connection formally. For numbers of the form 
\[
(a+b\sqrt{2})/\sqrt{2}^m, a,b,m \in \Z, m \ge 0
\]
we define $\sde$ as a minimal possible $m$, $m_{\min}$ such that the number can be written in the form $(a'+b'\sqrt{2})/\sqrt{2}^{m_{\min}},$ for $a',b'\in \Z.$ 

Let $u$ be an off-diagonal entry of a unitary $U$ over the ring $\Zr$ and let $\sde(|u|^2)=m.$ It was shown in Appendix B in \cite{es} that the optimal \Tcount for the circuit implementing the unitary $U$ can only be $m-2,m-1,m.$ It turns out that for given $|u|^2$ there always exists a circuit with optimal \Tcount $m-2.$ Indeed, by multiplying $U$ from right or left side by some power of  \T we can always achieve optimal \Tcount $m-2$ (see Appendix B in \cite{es}). From the other side, multiplying a unitary by powers of \T leaves the absolute value of its off-diagonal entries unchanged. 

To find a circuit implementing the unitary we apply the exact synthesis algorithm of \cite{es}, which produces a circuit with optimal number of \T gates. The algorithm is based on the fact that $\sde(|\cdot|^2)$ defines the complexity of the circuit that the unitary implements. The algorithm works by multiplying the unitary by $HT^l$ choosing $l$ to reduce $\sde(|\cdot|^2)$ of resulting unitary entries. The algorithm repeats this greedy approach until it reaches $\sde(|\cdot|^2)=3$ and then looks up the optimal circuit in a small database. More detailed description of the algorithm and the proof of \T optimality of produced circuits can be found in \cite{es}. 

Based on the discussion above we can restate the initial problem as: for fixed $m$ find a unitary with a minimal (but non--zero) off-diagonal entry $u$ such that $\sde(|u|^2)=m.$ The simplest approach is to go through all elements of the set 
\[
S_m:= \left\{  u \left| \begin{array}{c}
u \in \Zr,\: \sde(|u|^2) = m, \\
\exists v \in \Zr:|u|^2 + |v|^2 = 1 \\
\end{array} \right. \right\}
\]
and find its element with minimal absolute value. The condition $|u|^2 + |v|^2 = 1$ assures that there exist a unitary with off-diagonal entry $u.$ Therefore going through the set above is the same as going through all unitaries over the ring $\Zr.$ As a side note, the condition $\exists v \in \Zr : |u|^2 + |v|^2 = 1$ must be explicitly enforced because there exists $u \in \Zr$ such that $|u|<1,$ but $u$ is not an entry of any unitary over the ring $\Zr$. 

To iterate through all elements of $S_m$ it suffices to go through all $u \in \Zr$ with $\sde(|u|^2)=m$ and check the second condition $|u|^2 + |v|^2 = 1$. For $v$ expressed as $(v_0+v_1\w+v_2\w^2+v_3\w^3)/\sqrt{2}^\K$ the condition can be written as 
\[
\left|v_0+v_1\w+v_2\w^2+v_3\w^3\right|^2 = A + B\sqrt{2}, A,B \in \Z.
\]
The algorithm for solving such equations is known and is a part of several computer algebra systems. We use PARI/GP~\cite{pari} to check the existence of the solution for given $A,B.$ 

There is a systematic way to go through $u \in \Zr$ with $\sde(|u|^2)=m.$ Each $u$ can be described by five integers $a,b,c,d,\K$ and written as $(a+b\w+c\w^2+d\w^3)/\sqrt{2}^\K.$ The condition that $\sde(|u|^2)=m$ implies that we can chose $\K=\lceil m/2 \rceil.$ In addition, $u$ is required to be an entry of a unitary, therefore $|u|^2 + |v|^2 = 1$ for some $v.$ Multiplying the equality by $2^{\lceil m/2 \rceil}$ and collecting integer terms results in inequality
\[
 |a|^2 + |b|^2 + |c|^2 + |d|^2 \le 2^{\lceil m/2 \rceil}. 
\]
In summary, to go through all $u$ such that $\sde(|u|^2)=m$ it suffices to go through integers $a,b,c,d$ satisfying the inequality. The complexity of such a search procedure is exponential in $m.$ In the second part of this section we describe a search procedure that is still exponential, but more efficient and allows us to reach $m$ high enough to be interesting for our purposes. Note that to get the minimal absolute value $\delta$ of the off-diagonal entries found we need to consider $m$ that is in $O(\log(1/\delta))$; the complexity of both the simple and the improved search procedures is polynomial in $1/\delta.$
 
The improved search procedure uses additional information to shrink the search space. In particular we require that an upper bound $\ve$ on $|u|^2$ for given $m$ is provided as an input.  This bound can be taken to be the minimal value of $|u|^2$ for $m-1$. The procedure fails if the bound is too tight and an error message is returned, allowing the user to specify a less stringent error tolerance or increase the  value of $m$. 

Now we show how to use upper bound $\ve$ to shrink the search space. For our current purpose it is more convenient to represent $u$ as 
\[
( (a_0 + b_0\sqrt{2}) + i( a_1 + b_1\sqrt{2} ) ) / \sqrt{2}^\K.
\]
The bound $|u|^2\le\ve$ implies that $|a_0 + b_0\sqrt{2}|^2 + |a_1 + b_1\sqrt{2}|^2 \le 2^\K \ve$.  The savings are the most significant when $2^\K \ve \le 1/4$; in this case $a_j$ is uniquely defined by $b_j$ because $|a_0 + b_0\sqrt{2}| \le 1/2$ and $a_0$ must be equal to $\lfloor-b_0\sqrt{2} \rceil$. Our algorithm operates in this regime starting from $m\ge9$.

In the first stage of our search the algorithm builds list $L$ of triples $(a,b,|a + b\sqrt{2}|^2)$ such that $|a + b\sqrt{2}|^2 \le \sqrt{2^k\ve}$ and sorts it in ascending order by the third element. This allows the algorithm efficiently build the following list: 
\[
 L_{[0,\delta]} = \left\{ (a_0,b_0,c_0,d_0,r = |a_0 + b_0\sqrt{2}|^2 + |a_1 + b_1\sqrt{2}|^2 ), r \in [0,\delta] \right\}
\] 
for the chosen interval $[0,\delta].$ The algorithm again sorts the list in ascending order by the last element and finds the first element such that $( (a_0 + b_0\sqrt{2}) + i( a_1 + b_1\sqrt{2} ) ) / \sqrt{2}^\K$ can be an entry of the unitary. If it fails to find such an element then the algorithm restarts the procedure for a new list $L_{[\delta,2\delta]}.$ It keeps increasing list bounds either until it succeeds, or until it reaches the point where the lower bound for the list exceeds $2^\K \ve.$ In the second case, it reports that the initial bound was too tight.

\begin{table}[t!]
\caption{ Minimal absolute values of non-zero off-diagonal entries $u$ of unitaries with optimal \Tcount equal to $N_T.$ \label{tab:offd}}
  
 {\footnotesize
\begin{tabular}{|c|c|}
\hline
$N_T$ & $|u|$ \\
\hline
7 & 5.604e-02\\
10 & 2.145e-02\\
11 & 1.161e-02\\
13 & 8.207e-03\\
15 & 5.803e-03\\
17 & 4.104e-03\\
18 & 3.847e-03\\
19 & 1.202e-03\\
21 & 3.520e-04\\
23 & 2.489e-04\\
27 & 5.155e-05\\
31 & 2.578e-05\\
33 & 1.823e-05\\
34 & 1.709e-05\\
35 & 9.247e-06\\
37 & 1.564e-06\\
\hline
\end{tabular}
\begin{tabular}{|c|c|}
\hline
$N_T$ & $|u|$ \\
\hline
39 & 1.106e-06\\
41 & 7.818e-07\\
43 & 2.290e-07\\
45 & 1.619e-07\\
48 & 6.196e-08\\
51 & 2.371e-08\\
53 & 1.677e-08\\
56 & 2.658e-09\\
59 & 1.017e-09\\
63 & 2.107e-10\\
69 & 8.631e-11\\
71 & 6.103e-11\\
72 & 3.303e-11\\
73 & 1.542e-11\\
74 & 1.446e-11\\
76 & 1.022e-11\\
\hline
\end{tabular}
\begin{tabular}{|c|c|}
\hline
$N_T$ & $|u|$ \\
\hline
78 & 4.837e-12\\
79 & 4.614e-12\\
80 & 1.223e-12\\
83 & 8.103e-13\\
84 & 6.113e-13\\
85 & 4.875e-13\\
87 & 9.689e-14\\
88 & 9.082e-14\\
89 & 6.851e-14\\
92 & 3.864e-14\\
93 & 2.091e-14\\
94 & 1.330e-14\\
95 & 4.156e-15\\
98 & 3.840e-15\\
100 & 2.515e-15\\
 &  \\
\hline
\end{tabular}
}
\end{table}

\tab{offd} shows the results of running the described algorithm. For some values of $N_T$ (the optimal \Tchar--count) the minimal absolute value of off-diagonal matrix entries are not included in the table: for example, there are no values for the optimal \Tcount that equal eight and nine. This means that we can achieve smaller absolute values of off-diagonal entries using unitaries with optimal \Tcount seven than using unitaries with optimal \Tcount eight or nine. The same holds for all other intermediate values of optimal \Tcount  that are not included in~\tab{offd}.  The dependence of the optimal \Tcount on the minimal absolute value of off-diagonal matrix entries is plotted on \fig{optslope}.

\begin{figure}[t]
 \includegraphics[width=0.6\textwidth]{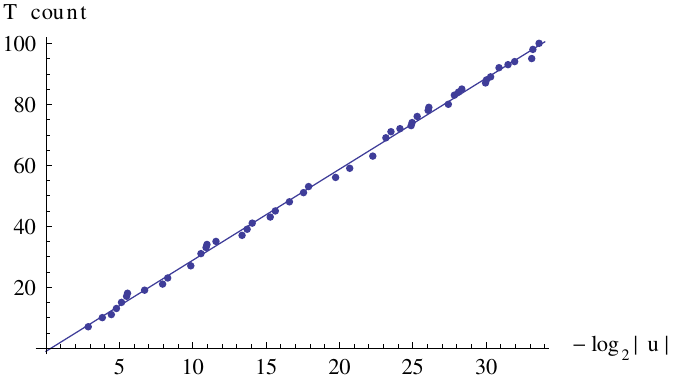}
\caption{ Here we show  the smallest absolute value, $|u|$, of the off-diagonal entiries of any unitary synthesized using the single--qubit Clifford and $T$ gate library as a function of the number of $T$--gates needed to attain the value of $|u|$. The data scales approximately as
$a\log_2(1/|u|)-1.064 $, where the value of $a$ that gives the least square error is 
$2.98$ and $a \in [2.95, 3.03]$ with probability 0.95. This is approximately $2.6$ times the value required to make a rotation of comparable size (or smaller) using composed gearbox circuits.\label{fig:optslope}}
\end{figure}

In summary, we have demonstrated a practical algorithm for finding single qubit unitaries drawn from the gate library consisting of single qubit Clifford gates and \T~that have the smallest possible absolute values of off-diagonal entries for values of the optimal \Tcount ranging from seven to one hundred. We see from the data in~\fig{optslope} and~\eq{compscale} that using ancillas and classical feedback for this task leads to improvement by approximately a factor of three in the \Tchar--count. To the best of our knowledge, this is the first example of a single qubit circuit synthesis task for which circuits including ancillas initialized to $\ket0$  and measurements with classical feedback require lower \Tcount in comparison to the optimal results involving only unitary operations.


\section{Conclusion}

Our work provides a new method for non--deterministically synthesizing small single qubit rotations.  We use this approach to construct a floating point representation of the rotation that can lead to substantial reductions in the \Tchar--count, \Tdepth and online \Tcount used to perform the rotations; furthermore, we show that the number of operations required to synthesize these rotations is less than lower bounds for the cost of synthesizing single qubit rotations using the \ClT gate library in cases where ancilla qubits are not used.

There are several avenues of future inquiry that are suggested by our work.
  Our results can be generalized by using different recursion relations at different depths in the recursive definition of our composed gearbox circuit.  Such generalizations allow modified versions of our circuits to closely approximate a much larger set of rotation angles and may lead to increased efficiency in certain cases.
Another important application of our work is in quantum simulation where implementing terms that are nearly negligible in a Trotter--Suzuki expansion is a common problem.   This application will be considered in subsequent work.

More generally, the non--deterministic circuits here could provide a large family of circuits that could be used to synthesize rotations.  Existing methods do not perform searches over non--deterministic circuits, such as those that we introduce here.  The addition of non--deterministic circuits, such as our gearbox circuits, as standard primitives for quantum circuit synthesis may lead to a much richer family of unitaries that can be synthesized using this approach and in turn lead to substantially reduced $T$ counts for synthesizing particular gates.

\acknowledgements
We would like to thank Chris Granade and Adam Paetznick for valuable comments on this work.  We also would like to acknowledge funding from USARO-DTO, CIFAR and NSERC.

Vadym Kliuchnikov is supported in part by the Intelligence Advanced Research Projects Activity (IARPA) via Department of Interior National Business Center Contract number DllPC20l66. The U.S. Government is authorized to reproduce and distribute reprints for Governmental purposes notwithstanding any copyright annotation thereon. Disclaimer: The views and conclusions contained herein are those of the authors and should not be interpreted as necessarily representing the official policies or endorsements, either expressed or implied, of IARPA, DoI/NBC or the U.S. Government.

\appendix
\section{Proof of \thm{fullcost}}\label{app:fullcost}

\begin{proofof}{\thm{fullcost}}
Let $x_d$ be a random variable that describes the number of times that $C^{\circ 1}(U)$ is applied before $C^{\circ d}(U)$ is successfully implemented.  Let $\eta_i$ and $\chi_i$ be independent random variables that are distributed as $x_{d-1}$ and let $N_d$ be the number of times that the final measurement in $C^{\circ d}(U)$ is applied.  Similarly to \cor{expect}, we see that $N_q$ is geometrically distributed with mean $1/P_q$ and variance $(1-P_q)/P_q^2$ for all $q\le d$.  The recursive definition of $C^{\circ d}(U)$ then implies that,
\begin{equation}
x_d = \sum_{i=1}^{N_d} (\eta_i + \chi_i) = \sum_{i=1}^{\infty} (\eta_i + \chi_i) \openone_{\{i\le N_d\}},
\end{equation}
where $\openone_{\{i\le N_d\}}=1$ if $i\le N_d$ and is zero otherwise.  We now substitute $\phi_i = \eta_i - \mathbb{E}(\eta_i)$ and $\xi_i = \chi_i - \mathbb{E}(\chi_i)$ in order to simplify our expressions for the expectation value and the variance of $x_d$ and obtain
\begin{equation}
x_d = \sum_{i=1}^{\infty} (\phi_i + \xi_i) \openone_{\{i\le N_d\}} + 2N_d\mathbb{E}(x_{d-1}),
\end{equation}
where $\mathbb{E}(\chi_i)=\mathbb{E}({x_{d-1}})=\mathbb{E}(\eta_i)$ because both random variables are distributed identically to $x_{d-1}$.
The expectation value of $x_d$ is then
\begin{eqnarray}
\mathbb{E}(x_d)&=\mathbb{E}\left(\sum_{i=1}^{\infty} (\phi_i + \xi_i) \openone_{\{i\le N_d\}}\right) + 2\mathbb{E}(N_d)\mathbb{E}(x_{d-1})\nonumber\\
&=2\mathbb{E}(N_d)\mathbb{E}(x_{d-1}),
\end{eqnarray}
where the last equation follows from the fact that $N_d$ is independent of $\phi_i$ and $\xi_i$ and $\mathbb{E}(\phi_i)=\mathbb{E}(\xi_i)=0$.

Now unfolding the recurrence relations, and using the fact that $\mathbb{E}(x_1)=\mathbb{E}(N_1)$ we have that
\begin{equation}
\mathbb{E}({x_d})=2^{d-1}\mathbb{E}(N_d)\cdots \mathbb{E}(N_1)= \frac{2^{d-1}}{P_d\cdots P_1}.\label{eq:expxd}
\end{equation}
Since ${C}^{\circ 1}(U)$ requires one application of $U$ and one application of $U^\dagger$, $n_d = 2 x_d$  and hence~\eq{expxd} implies
\begin{equation}
\mathbb{E}(n_d) = \frac{2^d}{P_1\cdots P_{d}},\label{eq:ndmean}
\end{equation}
as claimed.

The variance of $x_d$ can be expressed as
\begin{eqnarray}
\mathbb{V}(x_d)&= \mathbb{E}(x_d^2)-\left(\mathbb{E}(x_d)\right)^2\nonumber\\
&=\mathbb{E}\!\!\left(\!\!\left(\sum_{i=1}^{\infty} (\phi_i\! +\! \xi_i) \openone_{\{i\le N_d\}}\! +\! 2N_d\mathbb{E}(x_{d-1})\right)\!\!\!\left(\sum_{j=1}^{\infty} (\phi_j\! +\! \xi_j) \openone_{\{j\le N_d\}}\! +\! 2N_d\mathbb{E}(x_{d-1})\!\right)\!\right)\!\!\nonumber\\
 &\qquad-\!4\left(\mathbb{E}(N_d)\right)^2\!(\mathbb{E}(x_{d-1}))^2\nonumber\\
&=\mathbb{E}\left(\sum_{i=1}^{\infty} (\phi_i^2 + \xi_i^2) \openone_{\{i\le N_d\}})\right)+4\mathbb{E}(N_d^2)\mathbb{E}(x_{d-1})^2-4(\mathbb{E}(N_d))^2(\mathbb{E}(x_{d-1}))^2\nonumber\\
&=2\mathbb{V}(x_{d-1})\mathbb{E}(N_d)+4\mathbb{V}(N_d)(\mathbb{E}(x_{d-1}))^2,\label{eq:vrecur}
\end{eqnarray}
where we use the independence of $N_d$, $\phi_i$ and $\xi_i$, and the fact that the expectation value of $\phi_i$ and $\xi_i$ is zero to show the last equality.  Substituting~\eq{expxd} into~\eq{vrecur} gives
\begin{equation}
\mathbb{V}(x_d)=\frac{2\mathbb{V}(x_{d-1})}{P_d}+\frac{4^{d-1}(1-P_d)}{P_d^2\cdots P_1^2}.\label{eq:vrecur2}
\end{equation}

Unfolding the recursion relations in~\eq{vrecur2} and using the fact that $P_q$ is a monotonically increasing function of $q$ for all $\theta_0< \pi/4$, we find that
\begin{eqnarray}
\mathbb{V}(x_d)&=\frac{2^{d-1}\mathbb{V}(x_{1})}{P_d\cdots P_2}+\frac{4^{d-1}(1-P_d)}{P_d^2\cdots P_1^2}+\frac{2\cdot 4^{d-2}(1-P_{d-1})P_d}{P_d^2\cdots P_1^2}+\cdots\nonumber\\
&=\frac{2^{d-1}\mathbb{V}(x_{1})}{P_d\cdots P_2}+\sum_{j=1}^{d-1}\frac{2^{2d-j-1}(1-P_{d-j+1})\prod_{k=1}^{j-1}P_{d+1-k}}{P_d^2\cdots P_1^2}\nonumber\\
&\le\frac{2^{d-1}\mathbb{V}(x_{1})}{P_d\cdots P_2}+\sum_{j=1}^{d-1}\frac{2^{2d-j-1}(1-P_{d-j+1})}{P_d^2\cdots P_1^2}\nonumber\\
&\le\frac{2^{d-1}\mathbb{V}(x_{1})}{P_d\cdots P_2}+\sum_{j=1}^{\infty}\frac{2^{2d-j-1}(1-P_{d-j+1})}{P_d^2\cdots P_1^2}\nonumber\\
&\le \frac{2^{d-1}(1-P_1)}{(P_d\cdots P_1)P_1}+\frac{2^{2d-1}(1-P_1)}{P_d^2\cdots P_1^2}= \frac{2^{2d-1}(1-P_1)}{P_d^2\cdots P_1^2}\left(1+\frac{(P_d\cdots P_1)}{2^{d}P_1}\right).
\end{eqnarray}
The result of the theorem then follows from the fact that $\mathbb{V}(n_d)=4\mathbb{V}(x_d)$, similarly to~\eq{ndmean}.
\section{Cost Analysis}\label{app:cost}
\subsection{Cost Analysis of Floating Point Synthesis}
The $T$--count required to implement the circuit synthesis can easily be deduced from our prior discussions of the costs of the components of the floating point synthesis.   The circuit $C^{(d+1)}(U_m,C^{\circ D_1},\ldots, C^{\circ D_d})$ can be implemented using a mean $T$--count that scales (for some constant $C$) as
\begin{equation}
\mathbb{E}(T_{\rm count}) \approx 8\log_2(1/\delta)+1.14\log_2(10^\gamma)+C= O(\log_2(1/\theta)).\label{eq:tcountxpr}
\end{equation}
This says that for a fixed number of digits of precision ($\log_{10}(1/\delta)$), the cost of performing floating point synthesis is lower than that required by Selinger's method by nearly a factor of $4$ in the limit of small $\theta$ (large $\gamma$).  In fact, it is actually better than the best possible scaling that can be achieved using optimal ancilla--free single qubit synthesis, as is shown in~\sec{optimal}.

\eq{tcountxpr} can be verified using the following argument.  The only difference between performing  $C^{(d)}(U_m,C^{\circ D_1},\ldots,C^{\circ D_d})$ and $C^{(d)}(C^{\circ D_1},\ldots,C^{\circ D_d})$ is that two $U_m$ gates must be synthesized and one additional control must be added to the multiply--controlled $-iX$ gate ($\Lambda^{d}(-iX)$).
The $\Lambda^{d+1}(-iX)$ gate requires two more Toffoli gates to implement than $\Lambda^{d}(-iX)$ and so the extra control does not alter the scaling from that seen for $C^{(d)}(C^{\circ D_1},\ldots,C^{\circ D_d})$.  Furthermore, the inclusion of $U_m$ will actually boost the success probability of the circuit, which actually reduces the contribution of the $C^{\circ D_j}$ gates to the $T$--count.  \eq{tcountxpr} then follows from the fact that the mean $T$--count required to implement $C^{(d)}(C^{\circ D_1},\ldots,C^{\circ D_d})$ scales approximately as $1.14\log_2(1/\phi_{\rm rem})+C'=1.14\log_2(10^\gamma)+C''$ and the fact that the cost of implementing a rotation scales as $4\log_2(1/\delta)+C'''$ for constants $C',C''$ and $C'''$.

If the number of digits of precision required is not fixed, then floating point synthesis will provide a better $T$--count than Selinger's method given that
\begin{equation}
\delta \gtrapprox (10^{-\gamma})^{\frac{143}{200}}.
\end{equation}
This suggests that, for small rotation angles, extreme precision requirements will be needed for traditional circuit synthesis algorithms to have an advantage over our floating point synthesis method.  We therefore anticipate that in most circumstances our method will be favorable for implementing small rotations, if circuits with minimal $T$--count are required.

The $T$--depth required for our synthesis method and online $T$--counts required for our method are substantially smaller.  The expected $T$--depth for the floating point implementation is
$$\mathbb{E}(T_{\rm depth})=\mathbb{E}\left(2\max\left\{T_{\rm depth}(U_m),\max_j\{T_{\rm depth}(C^{\circ D_j})\}\right\}+T_{\rm depth}(\Lambda^{d+1}(-iX))\right).$$
As mentioned in \sec{compgearbox}, the majority of the operations in $C^{\circ D_j}$ can be thought of as ancilla preparations.
This means that any such ancilla preparation steps can be shifted offline and performed in parallel.  In essence, this reduces the depth of the circuit exponentially in exchange for a logarithmic increase in the circuit width.  This can easily be seen  using~\thm{fullcost}.  The only online operation that must be performed is $HTH$, which according to~\lem{ancillaprob}, will  only have to be performed a constant number of times before $C^{\circ D_j}$ is implemented with high probability.  This implies that
\begin{equation}
\mathbb{E}(T_{\rm depth})(C^{\circ D_j})\le D_j + K'\in O(\log\log(1/\theta))\in O(\log\log(10^\gamma)),
\end{equation}
where $K'\approx 12$ is a constant that arises from having to repeat the online step a fixed number of times.

The synthesis of $U_m$ using Selinger's method requires a $T$--depth that equals the $T$--count of the circuit.  This cost is
\begin{equation}
\mathbb{E}(T_{\rm depth}(U_m))=4\log_2(1/\delta)+K'',
\end{equation}
where $K''$ is a constant.  In our analysis this cost is assumed to be constant because the number of digits of precision, and in turn $\delta$, is assumed to be a constant.

The controlled $-iX$ gate $\Lambda^{d+1}(-iX)$ can be implemented using a depth $2\lfloor \log_2d +1\rfloor$ circuit~\cite{Sel13}.  This implies that the expected $T$--depth obeys, for some constant $K'''$,
\begin{eqnarray}
\mathbb{E}(T_{\rm depth})&\le 2 (\max_j D_j) + 8\log(1/\delta)+ 2\lfloor \log_2d +1\rfloor+K'''\nonumber\\ &\in O(\log\log(10^{\gamma})),
\end{eqnarray}
since $d\le \max_j D_j$ and $\max_j D_j \in \Theta(\log\log(10^\gamma))$.  Therefore the circuit depth varies doubly--logarithmically
with $\theta^{-1}$ and it is easy to see that the online cost follows a similar scaling.
This shows that another strong advantage of floating point synthesis is that it can easily exploit parallelism to reduce the time required to execute the circuits given that a fixed number of digits of precision are required.

\end{proofof}
\bibliographystyle{unsrt}

\begin{thebibliography}{10}

\bibitem{lloyd}
Seth Lloyd.
\newblock Universal quantum simulators.
\newblock {\em Science}, \textbf{273}:1073--1078, 1996.

\bibitem{LA97}
Daniel~S. Abrams and Seth Lloyd.
\newblock Simulation of many-body fermi systems on a universal quantum
  computer.
\newblock {\em Phys. Rev. Lett.}, 79:2586--2589, 1997.

\bibitem{WML+10}
H.~{Weimer}, M.~{M{\"u}ller}, I.~{Lesanovsky}, P.~{Zoller}, and H.~P.
  {B{\"u}chler}.
\newblock {A Rydberg quantum simulator}.
\newblock {\em Nature Physics}, 6:382--388, May 2010.

\bibitem{KW+11}
I.~{Kassal}, J.~D. {Whitfield}, A.~{Perdomo-Ortiz}, M.-H. {Yung}, and
  A.~{Aspuru-Guzik}.
\newblock {Simulating Chemistry Using Quantum Computers}.
\newblock {\em Annual Review of Physical Chemistry}, 62:185--207, 2011.

\bibitem{RWS12}
Sadegh Raeisi, Nathan Wiebe, and Barry~C Sanders.
\newblock Quantum-circuit design for efficient simulations of many-body quantum
  dynamics.
\newblock {\em New Journal of Physics}, 14(10):103017, 2012.

\bibitem{CMT+09}
Craig~R. Clark, Tzvetan~S. Metodi, Samuel~D. Gasster, and Kenneth~R. Brown.
\newblock Resource requirements for fault-tolerant quantum simulation: The
  ground state of the transverse ising model.
\newblock {\em Phys. Rev. A}, 79:062314, 2009.

\bibitem{UGS13}
Hao You, Michael~R. Geller, and P.~C. Stancil.
\newblock Simulating the transverse ising model on a quantum computer: Error
  correction with the surface code.
\newblock {\em Phys. Rev. A}, 87:032341, 2013.

\bibitem{DN06}
Christopher~M. Dawson and Michael~A. Nielsen.
\newblock The solovay-kitaev algorithm.
\newblock {\em Quantum Info. Comput.}, 6(1):81--95, January 2006.

\bibitem{HRC02}
A.~W. {Harrow}, B.~{Recht}, and I.~L. {Chuang}.
\newblock {Efficient discrete approximations of quantum gates}.
\newblock {\em Journal of Mathematical Physics}, 43:4445--4451, September 2002.

\bibitem{KMM12}
V.~{Kliuchnikov}, D.~{Maslov}, and M.~{Mosca}.
\newblock {Asymptotically optimal approximation of single qubit unitaries by
  Clifford and T circuits using a constant number of ancillary qubits}.
\newblock {\em arXiv:1212.0822}, December 2012.

\bibitem{Sel12}
P.~{Selinger}.
\newblock {Efficient Clifford+T approximation of single-qubit operators}.
\newblock {\em arXiv:1212.6253}, December 2012.

\bibitem{vbasis}
Alex Bocharov, Yuri Gurevich, and Krysta~M. Svore.
\newblock {Efficient Decomposition of Single-Qubit Gates into $V$ Basis
  Circuits}.
\newblock {\em arXiv:1303.1411}, March 2013.

\bibitem{CJ12}
N~Cody Jones, James~D Whitfield, Peter~L McMahon, Man-Hong Yung, Rodney~Van
  Meter, Al\`{a}n Aspuru-Guzik, and Yoshihisa Yamamoto.
\newblock Faster quantum chemistry simulation on fault-tolerant quantum
  computers.
\newblock {\em New Journal of Physics}, 14(11):115023, 2012.

\bibitem{DS12}
G.~{Duclos-Cianci} and K.~M. {Svore}.
\newblock {A State Distillation Protocol to Implement Arbitrary Single-qubit
  Rotations}.
\newblock {\em arXiv:1210.1980}, October 2012.

\bibitem{CL13}
A.~J. {Landahl} and C.~{Cesare}.
\newblock {Complex instruction set computing architecture for performing
  accurate quantum $Z$ rotations with less magic}.
\newblock {\em arXiv:1302.3240}, February 2013.

\bibitem{fstdist}
Cody Jones.
\newblock {Distillation protocols for Fourier states in quantum computing}.
\newblock {\em arXiv:1303.3066}, March 2013.

\bibitem{CGM+09}
Richard Cleve, Daniel Gottesman, Michele Mosca, Rolando~D. Somma, and David
  Yonge-Mallo.
\newblock Efficient discrete-time simulations of continuous-time quantum query
  algorithms.
\newblock In {\em Proceedings of the 41st annual ACM symposium on Theory of
  computing}, STOC '09, pages 409--416, 2009.

\bibitem{CW12}
A.~M. {Childs} and N.~{Wiebe}.
\newblock {Hamiltonian Simulation Using Linear Combinations of Unitary
  Operations}.
\newblock {\em Quantum Information and Computation}, 12:901--924, 2012.

\bibitem{FSG09}
Austin~G. Fowler, Ashley~M. Stephens, and Peter Groszkowski.
\newblock High-threshold universal quantum computation on the surface code.
\newblock {\em Phys. Rev. A}, 80:052312, 2009.

\bibitem{CJ13}
Cody Jones.
\newblock Low-overhead constructions for the fault-tolerant toffoli gate.
\newblock {\em Phys. Rev. A}, 87:022328, Feb 2013.

\bibitem{AMM+12}
M.~{Amy}, D.~{Maslov}, M.~{Mosca}, and M.~{Roetteler}.
\newblock {A meet-in-the-middle algorithm for fast synthesis of depth-optimal
  quantum circuits}.
\newblock {\em ArXiv e-prints}, June 2012.

\bibitem{NC00}
Michael~A. Nielsen and Isaac~L. Chuang.
\newblock {\em Quantum Computation and Quantum Information}.
\newblock Cambridge University Press, Cambridge U.K., Oct 2000.

\bibitem{prappr}
Vadym Kliuchnikov, Dmitri Maslov, and Michele Mosca.
\newblock {Practical approximation of single-qubit unitaries by single-qubit
  quantum Clifford and T circuits}.
\newblock {\em arXiv:1212.6964}, December 2012.

\bibitem{es}
Vadym Kliuchnikov, Dmitri Maslov, and Michele Mosca.
\newblock {Fast and efficient exact synthesis of single qubit unitaries
  generated by Clifford and T gates}.
\newblock {\em arXiv:1206.5236}, June 2012.

\bibitem{pari}
{ PARI, a computer algebra system, Online:
  \href{http://pari.math.u-bordeaux.fr}{http://pari.math.u-bordeaux.fr}}.

\bibitem{Sel13}
P.~{Selinger}.
\newblock {Quantum circuits of T-depth one}.
\newblock {\em Phys. Rev. A}, 87(4):042302, 2013.

\end{thebibliography}

\end{document}